\documentclass{article}
\usepackage{amsmath}
\usepackage{amssymb}
\usepackage{graphicx}
\usepackage{subcaption}
\usepackage{algorithm2e}
\usepackage{algpseudocode}
\usepackage{cite}
\usepackage{amsthm}  
\usepackage[margin=3.5cm]{geometry}
\usepackage{appendix}  
\newtheorem{proposition}{Proposition}
\newtheorem{theorem}{Theorem}
\newtheorem{assumption}{Assumption}
\newtheorem{example}{Example}
\newtheorem{corollary}{Corollary}
\newtheorem{definition}{Definition}

\title{External Incremental Delaunay Triangulation }
\author{Yifeng Cai}
\date{\today}

\begin{document}
\maketitle

\abstract{
This paper introduces a Delaunay triangulation algorithm based on the external incremental method. Unlike traditional random incremental methods, this approach uses convex hull and points as basic operational units instead of triangles.
Since each newly added point is outside the convex hull, there is no need to search for which triangle contains the point, simplifying the algorithm implementation.
The time complexity for point sorting is $O(n\log n)$, while the collective complexity for upper/lower tangent searches is proven to be $O(n)$.
For uniformly distributed point sets, empirical results demonstrate linear time $O(n)$ for full triangulation construction.
The overall time complexity remains $O(n\log n)$.
This paper details the algorithm's data structures, implementation details, correctness proof, and comparison with other methods.
}

\section{Introduction}
Delaunay triangulation is an important problem in computational geometry, with wide applications in computer graphics, geographic information systems, and other fields. Traditional random incremental methods use triangles as basic operational units, requiring the search for which triangle contains each newly added point. The external incremental method proposed in this paper avoids this search process by using convex hull as the basic operational unit, simplifying the algorithm implementation.

\section{Algorithm Overview}
The basic idea of the external incremental method is: By sorting the points with respect to their x-coordinates and y-coordinates, each newly added point is outside the current convex hull(the border of Delaunay triangulation), therefore it cannot be inside any existing triangle.
This property eliminates the need to search for which triangle contains the new point, simplifying the algorithm implementation.
We traverse the convex hull from the last added point, finding the upper tangent point counterclockwise and the lower tangent point clockwise.
The points between these two tangent points and closer to the new point are the points to be considered.

If the new point is inside the circumcircle of the triangle formed by the line segment between tangent points, we delete the edge of the triangle.
Delete an edge will expose two new edges, we need to check if the new point is inside the circumcircle of the triangle formed by these two new edges.
If it is, we delete the edge again.
This process continues until the new point is outside the circumcircle of all triangles.
At this point, we add the new point to the convex hull.
Link the upper and lower tangent points and all the exposed points between them to the new point, and a new Delaunay triangulation is formed.

\section{Data Structures}
\subsection{Array P}  
$P$ is a two-dimensional array used to store a set of two-dimensional points. It satisfies $x_i\leq x_j$ when $i < j$, and if $x_i=x_j$, $i<j$, then $y_i\leq y_j$.

\subsection{Index List L\_ch\_i}  
$L\_\text{ch}\_i$ is an index list whose elements point to points in $P$ that form a convex hull. It can be represented as: $L\_ch\_i = [ i_1, i_2, \ldots, i_m ]$.
Here, $i_k$ is the index of points in $P$ that form the convex hull.  

\subsection{Dictionary D\_tri}  
D\_tri is a dictionary representing the Delaunay triangulation. Its keys are pairs of point indices, and its values are lists of indices of the third points that form triangles with the key pair. Each list contains one or two integer values, which are indices of points that form triangles with the first two points. The specific form is:
$\{(i_1, i_2) \mapsto [i_3, i_4], \cdots,(i_5, i_6) \mapsto [i_7],\cdots\}$.

This indicates that in $P$, points with indices $i_1, i_2, i_3$ form a triangle, and points with indices $i_1, i_2, i_4$ also form a triangle. Specifically, the line segment $P[i_1]P[i_2]$ is an edge of the triangle, while points $P[i_3]$ and $P[i_4]$ are vertices opposite to this edge.

\section{Algorithm Implementation}
\subsection{Main Algorithm}
\begin{algorithm}[!h]
  \SetAlgoLined
  \KwData{Sorted point set P in 2D plane}
  \KwResult{Convex hull L of this point set, Delaunay triangulation}
  $i \leftarrow 2$
  
  \Repeat{L is not empty}{  
        $L\_ch\_i, D\_tri, i\_xmax \leftarrow \text{Try to create a triangle with P[0,1,...,i]}(P, i)$ 
        
        $i \leftarrow i+1$  
    }  
    \While{$i < len(P)$}{  
        $L, i\_xmax \leftarrow \text{Add point i to current Delaunay triangulation}(P,L\_ch\_i,D\_tri,i\_xmax,i)$\;  
        $i \leftarrow i + 1$\;  
    }  
    \Return{$L$}\; 
  \caption{External Incremental Method}
  \label{algo:sorted_incremental_delaunay}
\end{algorithm}

\begin{algorithm}[!h]
  \SetAlgoLined
  \KwData{P,A,D,a1,i\_p}
  \KwResult{L, i\_xmax, new Delaunay triangulation (stored in D\_tri, not explicitly returned)}
  Find upper tangent point a1 and lower tangent point a2 between P[i] and current convex hull
  
  Initialize N as empty list    \CommentSty{\%N acts as a stack}
    
  Push indices of all points between a1 and a2 into N

  L\_i=[A[a1]]   \CommentSty{\%L\_i records exposed points}
  
  find\_border\_while(P,D,N,L\_i,A[a1],i\_p)  \CommentSty{\%Implement erosion using while loop}

  Connect exposed points with p to update L, D\_tri
  \caption{Add point i to current Delaunay triangulation}
  \label{algo:add_point_i}
\end{algorithm}

\begin{algorithm}[!h]  
\SetAlgoLined  
\KwIn{Point set $P$, dictionary $D$, list $N$, list $L\_i$, index $i\_b$, index $i\_p$}  
\KwOut{Updated $D$, $N$, $L\_i$}  
\While{$N$ is not empty}{  
    \If{$(i\_b, N[-1])$ is an edge}{  
        $i\_d \leftarrow D[i\_b, N[-1]][0]$\;  
        \If{The point to be added is in its determined triangle}{  
            Delete this edge
            
            N.append(i\_d)
            
            \textbf{continue}
        }  
    }  
    i\_b $\leftarrow$N.pop()
    
    L\_i.append(i\_b)
}
\caption{find\_border\_while function}
\end{algorithm}

\subsection{Algorithm Illustration}
\begin{figure}[!htb]
    \begin{subfigure}{0.48\textwidth}
        \includegraphics[width=.95\textwidth]{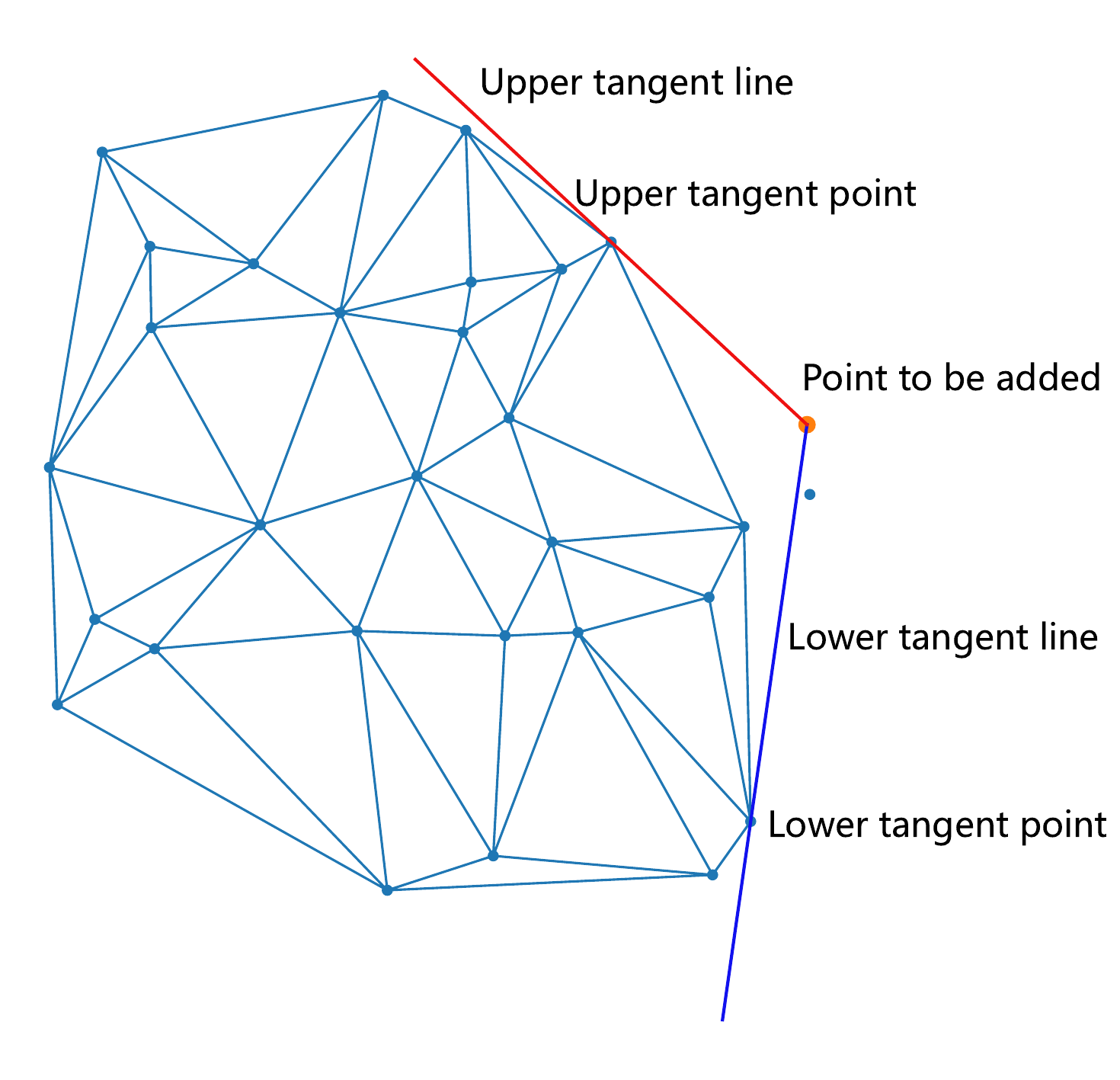}
        \caption{Find upper and lower tangent points of convex hull}
    \end{subfigure}
    \begin{subfigure}{0.51\textwidth}
        \includegraphics[width=.95\textwidth]{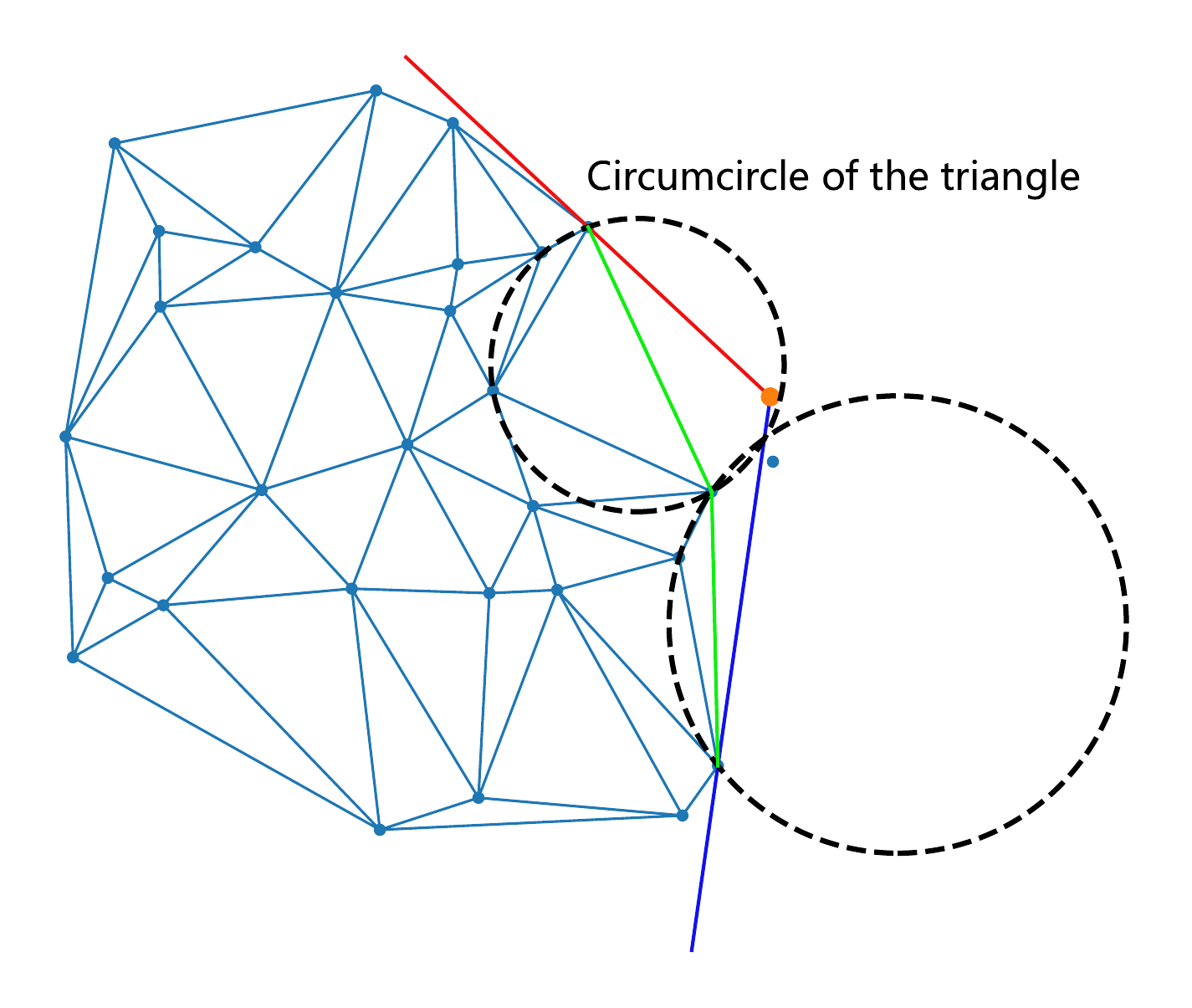}
        \caption{Determine if the new point is inside the circumcircle of the triangle\\
        formed by the line segment between tangent points}
    \end{subfigure}
    \begin{subfigure}{0.49\textwidth}
        \includegraphics[width=.95\textwidth]{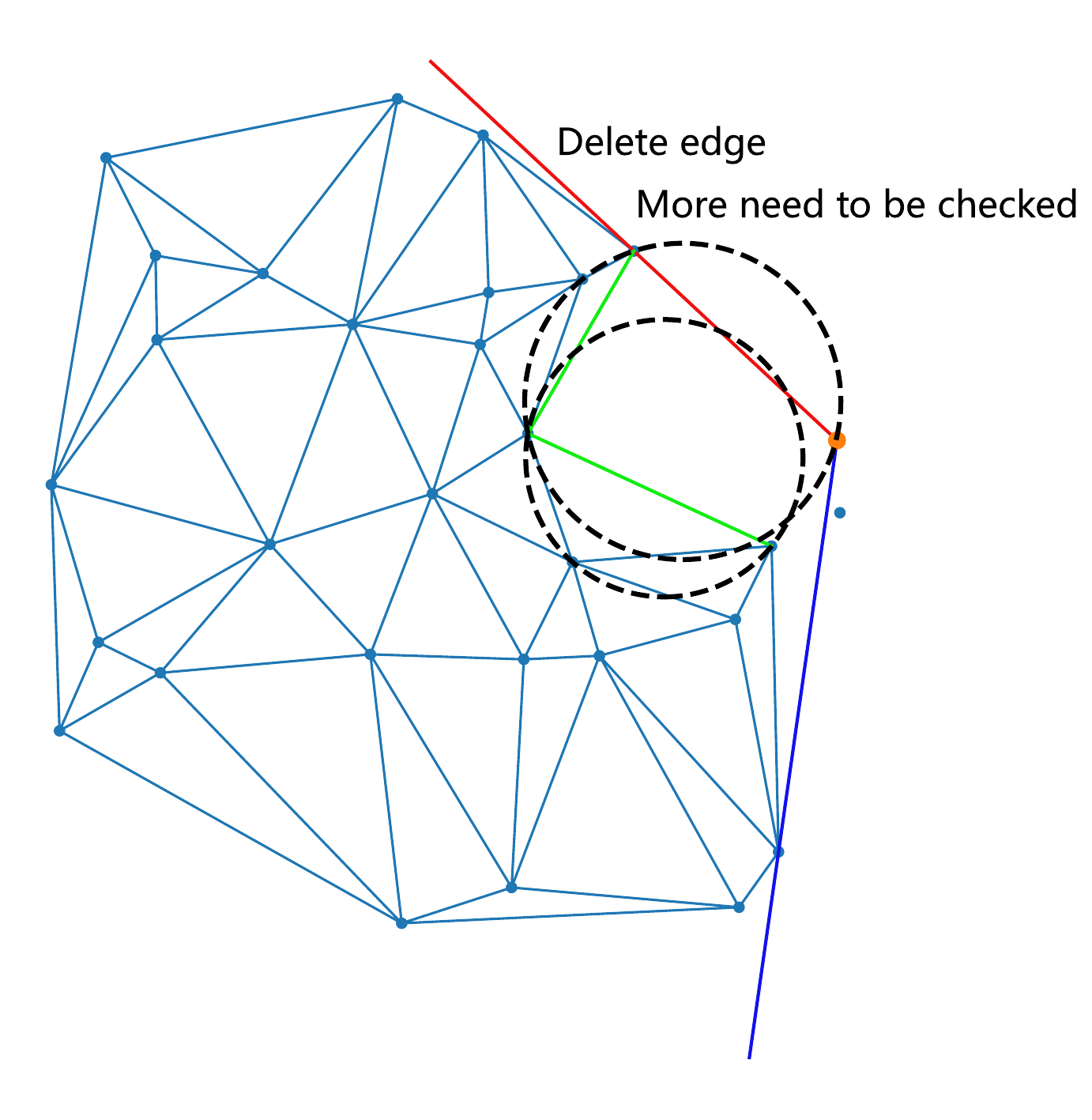}
        \caption{If an edge is deleted, two new edges need to be checked}
    \end{subfigure}
    \begin{subfigure}{0.5\textwidth}
        \includegraphics[width=.95\textwidth]{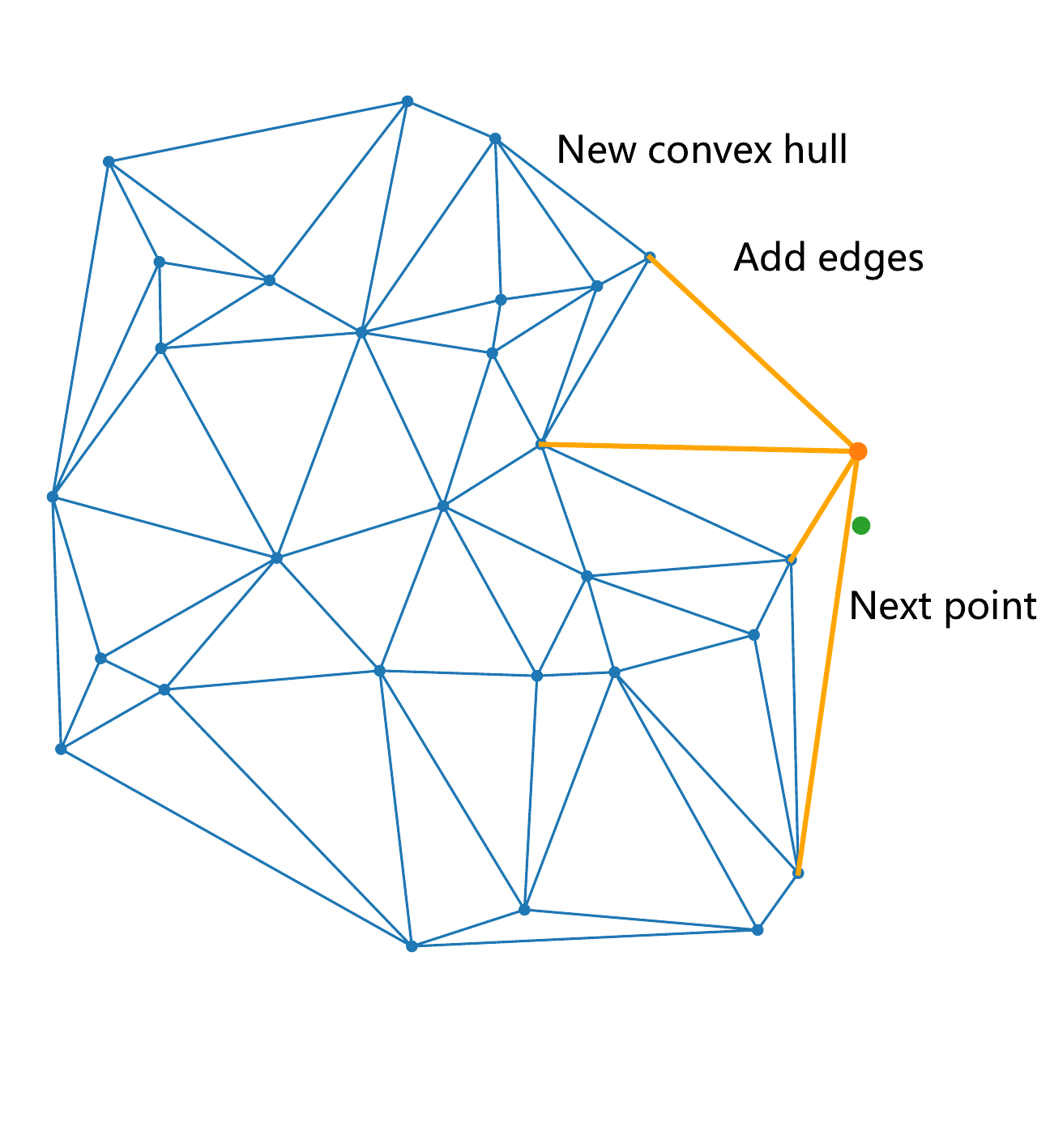}
        \caption{Connect exposed points with new point\\
        A new convex hull is formed, proceed to next iteration}
    \end{subfigure}
    \caption{Sorted incremental method for Delaunay triangulation}
\end{figure}

\clearpage
\section{Point in Circumcircle Test}
\textbf{Notation:} Let A be the point to be added, B and C be the points to be processed. D is the point corresponding to line segment BC, $\triangle BCD$ is a triangle. $\angle{A}=\angle{BAC}$, $\angle{D}=\angle{CDB}$

Note that point A to be added must be outside the convex hull. Therefore, it must be on the opposite side of point D corresponding to the line segment BC being processed. Also, note that neither $\angle A$ nor $\angle D$ is greater than $\pi$. With these two conditions, we can determine if point p is inside the circumcircle of $\triangle BCD$ by checking if $\angle{A}+\angle{D}$ is greater than $\pi$.

\begin{align}
   \angle{A}+\angle{D} & >  \pi\\
    \angle{A} & > \pi-\angle{D}\\
    \cos A& < -\cos D\\
    \cos A+\cos D& < 0
\end{align}
\begin{align}
    \frac{\textbf{AB} \cdot \textbf{AC}}{|AB| |AC|}+\frac{\textbf{BD} \cdot \textbf{CD}}{|BD| |CD|} & <  0\\
    |BD||CD|\textbf{AB} \cdot \textbf{AC} + |AB| |AC| \textbf{BD} \cdot \textbf{CD}
    & < 0 \label{eq:in_circle_inequality}
\end{align}

\begin{proposition}
    For $\triangle BCD$, if A and D are on opposite sides of BC, then if $\angle{A}+\angle{D}>\pi$, A is inside the circumcircle of BCD. If $\angle{A}+\angle{D}<\pi$, A is outside the circumcircle of BCD.
\end{proposition}

\begin{figure}[h]  
    \begin{subfigure}{.5\textwidth}  
        \centering  
        \includegraphics[width=.9\linewidth]{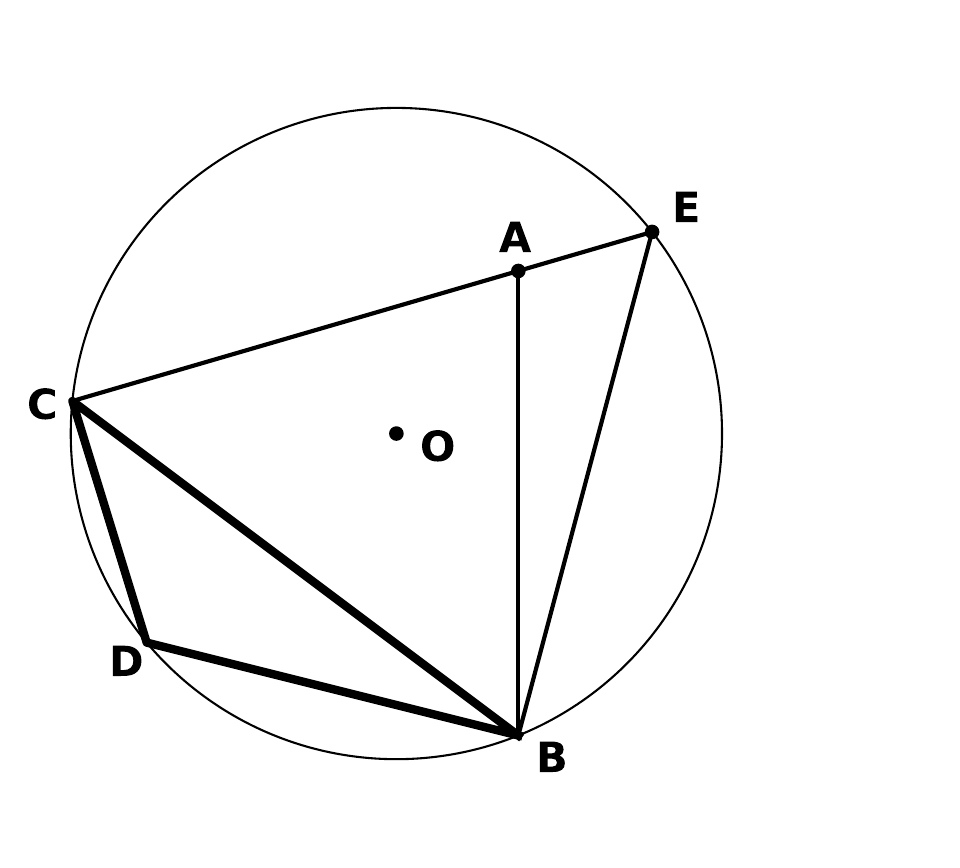}  
        \caption{A is inside the circumcircle of BCD}  
        \label{fig:A_in_BCD}  
    \end{subfigure}%
    \begin{subfigure}{.5\textwidth}  
        \centering  
        \includegraphics[width=.9\linewidth]{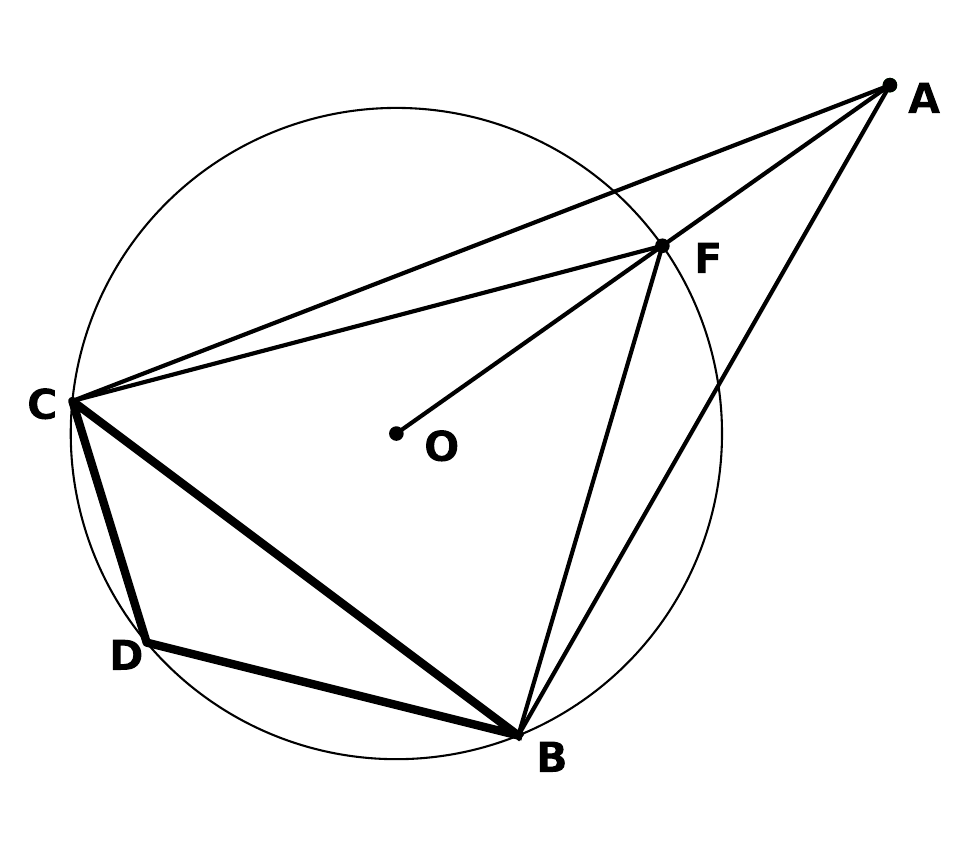}  
        \caption{A is outside the circumcircle of BCD}  
        \label{fig:A_out_BCD}  
    \end{subfigure}  
    \caption{A and D are on opposite sides of BC}  
    \label{fig:test}  
\end{figure}

\begin{proof}
    Let O be the circumcenter of $\triangle BCD$.
    
    If A is inside the circumcircle of BCD, extend CA to intersect the circle at point E. $\angle{BEC}+\angle{BDC}=\pi$. $\angle{BAC}=\angle{BEC}+\angle{EBA}>\angle{BEC}$. Therefore $\angle{BAC}+\angle{BDC}>\pi$.

    If A is outside the circumcircle of BCD, connect OA to intersect the circle at point F. $\angle{BFC}+\angle{BDC}=\pi$.
    
    $\angle{BFC}=\angle{BFO}+\angle{OFC}$, $\angle{BAC}=\angle{BAO}+\angle{OAC}$.
    
    $\angle{BFO}=\angle{BAO}+\angle{ABF}>\angle{BAO}$,
    $\angle{OFC}=\angle{OAC}+\angle{ACF}>\angle{OAC}$.

    Therefore $\angle{BAC}<\angle{BFC}$, $\angle{A}+\angle{D}<\pi$.
\end{proof}

Note that for $\triangle BCD$, when A and D are on the same side of BC, the above does not hold. The question of whether A is inside the circumcircle of BCD can be solved by computing the Cayley-Menger determinant, but this is a $5*5$ determinant. Clearly, Algorithm \ref{algo:sorted_incremental_delaunay} ensures that A and D are always on opposite sides of BC at each step, so we can use inequality \ref{eq:in_circle_inequality}. Generally, division is more prone to precision loss than multiplication, and some processor instruction sets (like some ARM architectures) don't even include division operations, so inequality \ref{eq:in_circle_inequality} is converted to a form using only multiplication.

For maximum performance, we can first check the signs of $\textbf{AB} \cdot \textbf{AC}$ and $\textbf{BD} \cdot \textbf{CD}$. If both are positive or both are negative, we can stop. If one is positive and one is negative, we can square both sides to eliminate square root operations for computing $|CD|$ etc.

\section{Correctness Proof}
This can be clearly seen through the dual graph of Delaunay triangulation, the Voronoi diagram.
\begin{theorem}
    In Algorithm \ref{algo:sorted_incremental_delaunay}, inequality \ref{eq:in_circle_inequality} always holds. That is, point A to be added is always on the opposite side of point D corresponding to line segment BC being processed.
\end{theorem}

\begin{theorem}
    Through a series of edge flips, any two triangulations of the same planar point set can be transformed into each other.\cite{lawson1972transforming}
\end{theorem}

The Delaunay graph is dual to the Voronoi diagram. Each Delaunay triangle corresponds to a Voronoi region. Edges in the Delaunay graph correspond to edges in the Voronoi diagram. Points in the Delaunay graph correspond to sites in the Voronoi diagram.
By the property of Voronoi diagrams: a site $p_i$'s Voronoi region $\text{VR}(p_i)$ is open if and only if $p_i$ is on the convex hull of point set $P$. Therefore, in Algorithm \ref{algo:sorted_incremental_delaunay}, after adding each new point, some Voronoi sites that were previously on the convex hull become interior points. Their Voronoi regions become closed. By $\text{VR}(p_i)=\{q|~|qp_i| \leq |qp_j|,~ \forall j \neq i,~ p_i, p_j \in P\}$, $\text{VR}(p_i)=\bigcap\limits_{i \neq j} h(p_i, p_j)$, $h(p_i, p_j)=\{q|~|qp_i| \leq |qp_i|\}$. That is, changes to a site's Voronoi region are caused by the newly added point. Therefore, there must exist a Voronoi edge whose associated sites are the original open Voronoi region's site and the newly added site. That is, new edges in the Delaunay graph must be connected to the newly added point.

\section{Analysis and Comparison}
The point addition operation in this method is similar to the random incremental method\cite{bernard1985delaunay}. However, compared to the random incremental method, this method does not use a fictitious large outer triangle, thus eliminating a series of operations and assumptions related to this fictitious triangle. Additionally, this method does not involve determining which triangle contains a point. Therefore, there is no need to dynamically maintain a triangle search tree during construction.

\appendix
\section{Convex Hull}
\label{chapter:convex_hull}

\subsection{Sorted(External) Incremental Method}
\begin{assumption}
In this paper, the convex hull is represented by a list $L$, where the elements in $L$ are the indices of points in $P$. Moreover, the convex hull lies on the left side of $L[i]L[i + 1]$, which means that as $i$ increases, it rotates counterclockwise along the convex hull.
\end{assumption}

\begin{example}
$L = [0, 1, 2]$ represents a triangle formed counterclockwise by $P[0]$, $P[1]$, and $P[2]$.
\end{example}

\begin{algorithm}[h]
  \SetAlgoLined
  \KwData{A set of points $P$ sorted by the $x$-axis in a two-dimensional plane}
  \KwResult{The convex hull $L$ of this set of points
  }

  $i \leftarrow 2$
  
  \Repeat{$L$ is non-empty}{  
        $L, i\_xmax \leftarrow \text{Try to create a triangle with } P[0, 1, ..., i] (P, i)$ 
        
        $i \leftarrow i + 1$  
    }  
    \While{$i < len(P)$}{  
        $L, i\_xmax \leftarrow \text{Add } P[i] \text{ to the convex hull } L (L, P[i], i\_xmax)$\;  
        $i \leftarrow i + 1$\;  
    }  
    \Return{$L$}\; 
  \caption{Sorted(External) Incremental Method}
  \label{algo:sorted_incremental_convex_hull}
\end{algorithm}

\begin{figure}[h]
    \centering
    \includegraphics[width=0.6\textwidth]{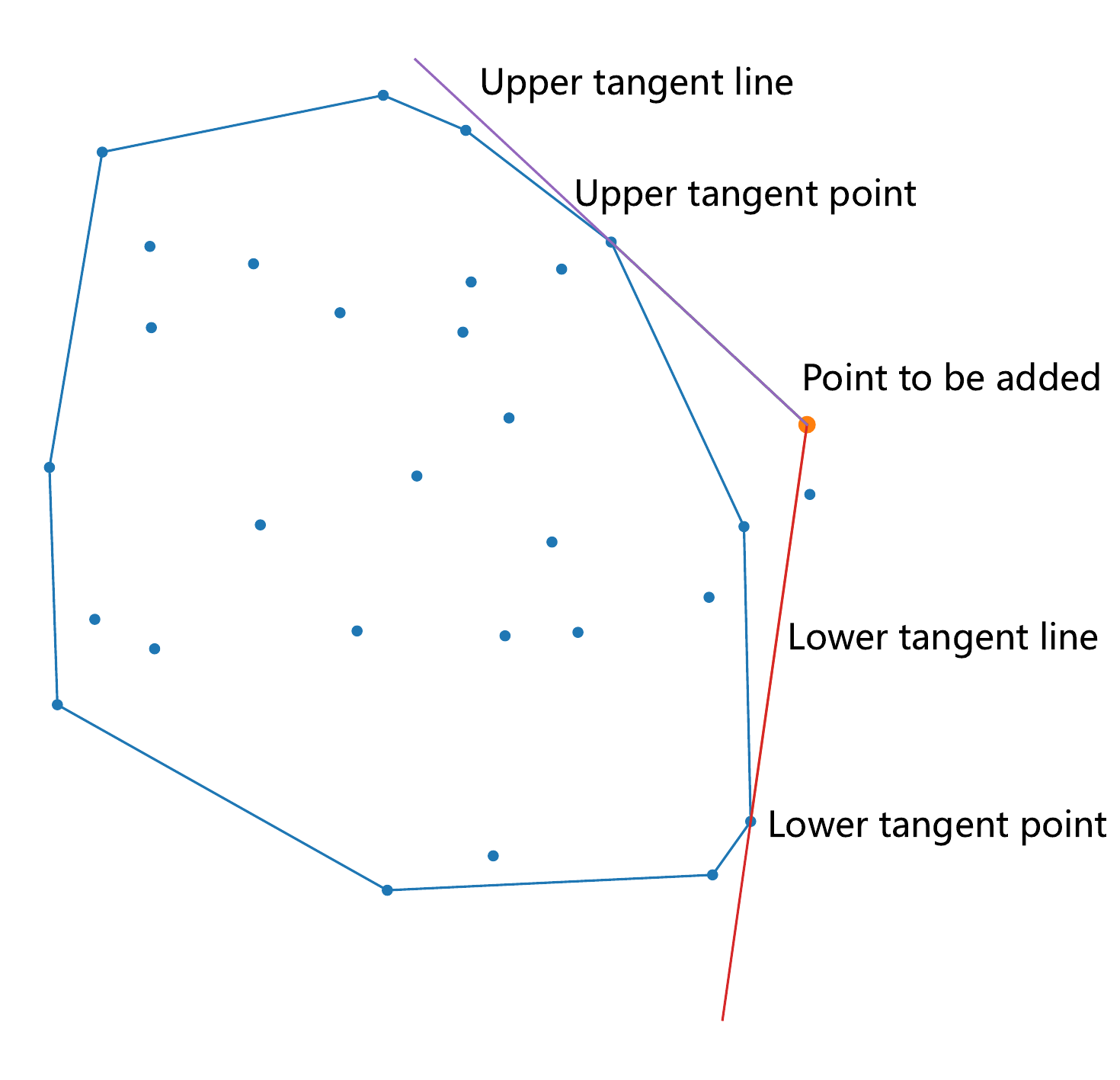}
    \caption{Schematic diagram of the sorted incremental method for the convex hull}
\end{figure}
\begin{algorithm}[h]  
\SetAlgoLined  
\KwIn{A list $A$ of points, a point $p$, and an index $a1$}  
\KwOut{An updated list $A$ with $p$ inserted at the appropriate position, and the new index $a1$}  
\BlankLine  
\If{$(\text{ToLeft}(p, A[a1 - 1], A[a1]) > 0)$}{  
    $a2 \leftarrow a1 - \text{length}(A) + 1$\;  
}  
\Else{  
    $a2 \leftarrow a1 - \text{length}(A)$\;  
    \While{$(\text{ToLeft}(p, A[a1 - 1], A[a1]) \leq 0)$}{  
        $a1 \leftarrow a1 - 1$\;  
    }  
}  
\While{$(\text{ToLeft}(p, A[a2], A[a2 + 1]) \leq 0)$}{  
    $a2 \leftarrow a2 + 1$\;  
}  
\If{$(a2 == 0)$}{  
    \Return{$A[:a1 + 1] + [p]$, $a1 + 1$}\;  
}  
\Return{$A[:a1 + 1] + [p] + A[a2:]$, $a1 + 1$}\;  
\caption{Add\_point\_to\_convex\_hull}
\label{alg:Add_a_point_to_convex_hull} 

\end{algorithm} 

\subsection{Derivation of the Time Complexity Formula}
To comprehensively list the formula for calculating its time complexity, the pseudocode of the entire program has to be presented.

\begin{theorem}
\label{theo:convex_initial}
If the first triangle is constructed based on the first $k$ points, then a total of $k - 2$ $ToLeft()$ operations are used at this time.
\end{theorem}

\begin{proof}
The algorithm \ref{algo:sorted_incremental_delaunay} first uses a $ToLeft()$ operation to determine whether $P[2]$ is on the left or right side of $\textbf{P[0]P[1]}$.
    \begin{enumerate}
        \item[Case 1:] If $P[0]$, $P[1]$, and $P[2]$ are collinear, then continue to try the $ToLeft()$ operation to determine whether $P[3]$ is on the left or right side of $\textbf{P[0]P[2]}$.
        \item[Case 2:] If $P[0]$, $P[1]$, and $P[2]$ are non-collinear, then $P[2]$ is on the left or right side of $\textbf{P[0]P[1]}$, and the first triangle is successfully constructed, and the algorithm ends.
    \end{enumerate}
    In summary, each $ToLeft()$ operation increases the index of the point by 1.
    
    Therefore, if the first triangle is constructed based on the first $k$ points, then a total of $k - 2$ $ToLeft()$ operations are used at this time.
\end{proof}

\begin{definition}[Tangent Line]
When a point is outside the convex hull, the line formed by it and a point on the convex hull, and all points on the convex hull are in a certain half-plane determined by this line, this line is called a tangent line.
\end{definition}
\begin{definition}[Tangent Point]
The point on the convex hull that is connected to the point outside the convex hull to form a tangent line is called a tangent point.
\end{definition}

In implementation, let the point $p$ outside the convex hull form a line $pA[a1]$ with the point $A[a1]$ on the convex hull, and both $pA[a1 + 1]$ and $pA[a1 - 1]$ are on the same side of it, that is, \[ToLeft(p, A[a1], A[a1 + 1]) \times ToLeft(p, A[a1], A[a1 - 1]) \geq 0\]
If $p$, $A[a1]$, and $A[a1 \pm 1]$ are collinear, then the point farther from $p$ is taken as the tangent point.

In the algorithm \ref{alg:Add_a_point_to_convex_hull}, since it is guaranteed that $p$ is on the right side of $A[i\_xmax]$ and rolls along the edge of the convex hull, the tangent point can be determined when the sign of $ToLeft()$ changes.

\begin{theorem}
When adding a new point $P[i]$ to the convex hull, if only 2 $ToLeft()$ operations are performed, no points will be deleted. If $m_i \geq 3$ $ToLeft()$ operations are performed, then $m_i - 3$ or $m_i - 2$ points on the convex hull will be deleted.
\end{theorem}

\begin{proof}
Let $m$ be the number of $ToLeft()$ operations performed.
According to the algorithm \ref{alg:Add_a_point_to_convex_hull}, first perform a $ToLeft()$ operation to determine whether $pA[a1]$ is the lower tangent line.
\begin{enumerate}
    \item When $pA[a1]$ is already the lower tangent line, then it must not be the upper tangent line. To find the upper tangent line, start directly from the next point of $A[a1]$, $a2 \leftarrow a1 + 1$, that is, $A[a2]$. Thereafter, if the current $pA[a2]$ is not the upper tangent line, then $a2 \leftarrow a2 + 1$. Since only the point that finally forms the upper tangent line is retained, the other points that are verified by the $ToLeft()$ operation not to be the upper tangent line are deleted. Each additional $ToLeft()$ operation, $a2 \leftarrow a2 + 1$, means that one more point on the convex hull is deleted. If $pA[a1]$ is the lower tangent line and $pA[a1 + 1]$ is the upper tangent line, then $m_i = 2$, and no points are deleted. On this basis, each time $a2 \leftarrow a2 + 1$, $m_i \leftarrow m_i + 1$, and one point on the original convex hull is deleted.

    \item On the other hand, if $pA[a1]$ is not the lower tangent line, then $a1 \leftarrow a1 - 1$ is required, and the rest is the same as above, but this time the lower tangent line is being found. However, it cannot be determined whether $pA[a1]$ is the upper tangent line, so it is still necessary to determine whether $pA[a1]$ is the upper tangent line, that is, if not, then $a2 \leftarrow a2 + 1$, and the rest is the same as above. So at this time, $m_i - 3$ points are deleted.
\end{enumerate}

In summary, it is proven.

\end{proof}

\begin{assumption}
    
Let the number of times that $P[i]A[i\_xmax]$ is not the lower tangent line during the construction of the convex hull be $h$.
Define the indicator function $I[i]$ as:

\[
I[i] =
\begin{cases}
0, & \text{if } P[i]A[i\_\text{xmax}] \text{ is the lower tangent line} \\
1, & \text{if } P[i]A[i\_\text{xmax}] \text{ is not the lower tangent line}
\end{cases}
\]
Then when $P[i]$ is added to the convex hull, $m_i - 2 - I[i]$ points are deleted. $\sum I[i] = h$
\end{assumption}

\begin{theorem}
In the algorithm \ref{alg:Add_a_point_to_convex_hull}, $k$ points are added to the convex hull to construct a convex hull. The number of times that $P[i]A[i\_xmax]$ is not the lower tangent line during the construction process is $h$. After the algorithm ends, the change in the number of points on the convex hull is $\Delta l$. Then a total of $3k + h - \Delta l$ $ToLeft()$ operations are performed.

\end{theorem}
\begin{proof}
Let $\sum m_i$ be the total number of $ToLeft()$ operations performed,
$\sum I[i] = h \leq k$,
    \begin{align}
  k - \sum (m_i - 2 - I[i]) & = \Delta l \\
  k - \sum m_i + 2k + h & = \Delta l \\
  \sum m_i = 3k + h - \Delta l
    \end{align}
\end{proof}
\begin{theorem}
Suppose there are a total of $n$ points. The algorithm \ref{algo:sorted_incremental_convex_hull} needs the first $k_1$ points to construct the first convex hull. After all points participate in the construction, there are a total of $l$ points on the convex hull. Then a total of $3n - 2k_1 - l + h + 1$ $ToLeft()$ operations are performed. $h$ depends on the data distribution, and $h \leq n - k_1$.
\end{theorem}
\begin{proof}
When the first convex hull is constructed based on the first $k_1$ points, this convex hull is a triangle, and there are a total of 3 points on the convex hull. According to \ref{theo:convex_initial}, $k_1 - 2$ $ToLeft()$ operations have already been performed at this time.
The remaining number of points is $k_2 = n - k_1$, so $\sum m_i = 3k_2 + h - \Delta l$ $ToLeft()$ operations are performed.
Finally, there are a total of $l$ points on the convex hull, so $\Delta l = l - 3$.
Therefore, the total number of $ToLeft()$ operations performed is:
    \begin{align}
            & 3k_2 + h - (l - 3) + k_1 - 2 \\
        = & k_1 + 3k_2 + h - l + 1 \\
        = & 3n - 2k_1 - l + h + 1
    \end{align}
\end{proof}
\begin{corollary}
A set of points sorted by the $x$-coordinate can be converted into a convex hull in $O(n)$ time.
\end{corollary}
In terms of probability, the probability that three points are collinear is almost 0, but it can actually happen. In most cases, $k_1 = 3$.
\begin{corollary}
In most cases, the algorithm \ref{algo:sorted_incremental_convex_hull} requires $3n - l + h - 5$ $ToLeft()$ operations.
\end{corollary}
Note that since $h < n$, we have $3n - l + h - 5 < 4n$, which implies a time complexity of $O(n)$. Furthermore, the collective complexity for upper/lower tangent searches has been proven to be $O(n)$.
\subsection{Comments}
First sort the points along the $x$-axis, and then construct the convex hull. However, each step of my algorithm returns the convex hull, rather than two monotonic chains (upper and lower convex hulls). Also, note that if $pL[i\_xmax]$ is the lower tangent line, it must not be the upper tangent line. Therefore, $n - h$ $ToLeft()$ operations are saved.

Compared with the currently recognized fastest Graham's algorithm \cite{graham1972efficient}, generally speaking, Graham's algorithm requires 1 $ToLeft()$ operation for points that will be on the convex hull and 2 operations for points that are not. However, \textbf{it needs to calculate the angles of all points relative to an extreme point}, and ignoring the constant term, it is $2n - l$. For the worst case of \ref{algo:sorted_incremental_convex_hull}, taking $h = n - k_1$, it is $4n - 3k_1 - l + 1$. That is, in the worst case, it is $2n$ more operations than Graham's algorithm. In the best case, it is only $n$ more operations. Therefore, the speed comparison between the two algorithms depends on which is faster between \textbf{calculating 1 angle} and \textbf{(1 + p) $ToLeft()$ operations} where $0 \leq p \leq 1$, which varies for different programming languages and execution environments.  

\end{document}